\documentclass[10pt,A4paper,conference]{IEEEtran}
\usepackage{graphicx}
\usepackage{ifthen, calc, amssymb, amsmath, amsfonts}

\def\calS{{\cal S}}
\def\calX{{\cal X}}
\def\calY{{\cal Y}}
\def\calZ{{\cal Z}}

\def\bx{\textbf{{x}}}
\def\bX{\textbf{{X}}}
\def\by{\textbf{{y}}}
\def\bY{\textbf{{Y}}}
\def\bz{\textbf{{z}}}

\def\bE{\textbf{{E}}}
\def\b1{\mathbf{1}}

\begin{document}

\newtheorem{theorem}{\bf Theorem}
\newtheorem{lemma}{\bf Lemma}
\newtheorem{corollary}[theorem]{\bf Corollary}
\newtheorem{proposition}[theorem]{\bf Proposition}
\newtheorem{definition}{\bf Definition}
\newtheorem{example}{\bf Example}

\def\zproof{\noindent{\itshape Proof of }}
\def\endzproof{\hspace*{\fill}~\QED\par\endtrivlist\unskip}

\title{Markov Lemma for Countable Alphabets}


\author{\authorblockN{Siu-Wai Ho}
\authorblockA{Institute for Telecommunications Research\\
University of South Australia\\
Australia\\
Email: siuwai.ho@unisa.edu.au}
 }
%
\maketitle
\begin{abstract}
Strong typicality and the Markov lemma have been used in the proofs of several multiterminal source coding theorems. Since these two tools can be applied to finite alphabets only, the results proved by them are subject to the same limitation. Recently, a new notion of typicality, namely unified typicality, has been defined.  It can be applied to both finite or countably infinite alphabets, and it retains the asymptotic equipartition property and the structural properties of strong typicality.  In this paper, unified typicality is used to derive a version of the Markov lemma which works on both finite or countably infinite alphabets so that many results in multiterminal source coding can readily be extended. Furthermore, a simple way to verify whether some sequences are jointly typical is shown.
\end{abstract}

\section{Introduction}
The Markov lemma was first used by Berger \cite{Berger} to extend
multiterminal source coding theory.  It has been used in the
achievability part of the coding theorems in source coding with side
information \cite[Section~15.8]{bk:Cover}, rate distortion with side
information \cite[Section~15.9]{bk:Cover}, channel coding with
side information \cite[Section~6.2]{bk:Kramer}, a large class of
multiterminal noiseless source coding problems \cite{HK}, etc. The
different versions of the Markov lemma given in
\cite{Berger}--\cite{HK} have the same limitation that all of them
cannot be applied to countably infinite alphabets because they are
based on strong typicality \cite{Berger}\cite{bk:Raymond}. Note
that the Markov lemma for Gaussian sources has been shown in
\cite{Oo}.



Recently, Ho and Yeung have defined a new notion of typical
sequences, called unified typicality, which works for countable
alphabets\footnote{Countable alphabet means an alphabet which
can be finite or countably infinite} \cite{UnifiedTypicality}.  Unified
typicality retains the asymptotic equipartition property and the
structural properties of strong typicality \cite{UnifiedTypicalityJ}. We
will further show in this paper that unified typicality can give a
version of the Markov lemma for countable alphabets, which can be
used to extend the achievability parts of the aforementioned coding
problems. Also, the new Markov lemma further supports that unified typicality is a right notion for generalizing strong
typicality to countable alphabets.

%

In order to show that some sequences are jointly weakly typical, we
need to show $2^k - 1$ nonnegative quantities in
\cite[(15.24)]{bk:Cover} sufficiently small for a problem with $k$
random variables. It seems that unified typicality suffers the same
trouble. In this paper, we will demonstrate a simple method which
requires to show only two nonnegative quantities sufficiently small in
order to show jointly unified typical.



In the next section, we introduce unified typicality and some notations. In Section~\ref{se:TheMarkovL}, the Markov lemma which works on both finite or countably infinite alphabet is shown, and its consequences are discussed. Then some useful lemmas and the trick to ease the verification of jointly unified typical sequences are shown in Section~\ref{se:somelemmas} before the new Markov lemma is proved in Section~\ref{se:proof}. In this paper, the base of the logarithm is $2$.

\section{Unified Typicality \label{se:UnifiedT}}
Consider some countable alphabets $\calX$, $\calY$ and $\calZ$.
For any sequences $\by = (y_1, \ldots, y_n) \in \calY^n$, we say
that a sequence of random variables $\textbf{X}=(X_1, X_2, ...,
X_n) \in \calX^n$ is drawn $\sim\prod_{i} p(x_i|y_i)$ if $X_i$ are
independent and
\begin{eqnarray}
\Pr\{\bX = \bx\} =
\prod_{i=1}^n p(x_i|y_i),
\end{eqnarray}
where $\bx = (x_1, \ldots, x_n) \in \calX^n$.  Let $\bz = (z_1,
\ldots, z_n) \in \calZ^n$. We call $Q_{XYZ} = \{q(xyz)\}$ the
\emph{empirical distribution} of the sequences $(\bX, \by, \bz)$,
where $q(xyz) = n^{-1}N(x, y, z; \textbf{X}, \by, \bz)$ and $N(x,
y, z; \textbf{X}, \by, \bz)$ is the number of occurrences of $(x, y,
z)$ in the sequences $(\textbf{X}, \by, \bz)$. Note that $Q_{XYZ}$
is also called the \emph{type} of $(\textbf{X}, \by, \bz)$
\cite{CsiszerandKorner} and $Q_{XYZ}$ is a random variable as
$\bX$ is random. The marginal distribution $\{q(xy)\}$ is denoted
by $Q_{XY}$ and the other marginal distributions of $Q_{XYZ}$
and $P_{XYZ} = \{p(xyz)\}$ are defined in a similar fashion. We
use $X - Y - Z$ to denote a Markov chain with respect to
$P_{XYZ}$, i.e., $p(xyz) = p(x|y)p(yz)$ for all $x$, $y$ and $z$.
Now, we use the Kullback-Leibler divergence $D( \cdot|| \cdot)$
and entropy $H(\cdot)$ (see e.g., \cite{bk:Cover}\cite{bk:Raymond}) to define unified typicality \cite{UnifiedTypicality}. We always assume $H(P_{XYZ}) < \infty$.

\medskip
\begin{definition}
\label{df:UJtypicality}  The unified jointly typical set
$U^n_{[XYZ]\gamma}$ with respect to $P_{XYZ}$ is the set of
sequences $(\textbf{x}, \textbf{y}, \bz) \in \mathcal{X}^n \times
\mathcal{Y}^n \times \mathcal{Z}^n$ such that
\begin{eqnarray}
D({Q'}_{XYZ}||P_{XYZ}) +
|H({Q'}_{XYZ})-H(P_{XYZ})| + \nonumber\\
|H({Q'}_{XY})-H(P_{XY})| +
|H({Q'}_{YZ})-H(P_{YZ})| + \nonumber\\
|H({Q'}_{XZ})-H(P_{XZ})| +
|H({Q'}_{X})-H(P_{X})| + \nonumber\\
|H({Q'}_{Y})-H(P_{Y})| +
|H({Q'}_{Z})-H(P_{Z})|\leq \gamma,
\label{eq:UJtypicality}
\end{eqnarray}
where $Q'_{XYZ} = \{q'(xyz)\}$ is the empirical distribution of
$(\bx, \by, \bz)$ with $q'(xyz) = n^{-1}N(x, y, z; \textbf{x}, \by,
\bz)$.
\end{definition}
The definition of $U^n_{[YZ]\gamma}$ is similar to
$U^n_{[XYZ]\gamma}$ with $D({Q}_{XYZ}||P_{XYZ})$ replaced
by $D({Q}_{YZ}||P_{YZ})$ and all the absolute values involving $X$
being dropped.

\section{Main Results \label{se:MainResults}}
\subsection{The Markov Lemma \label{se:TheMarkovL}}
The Markov lemma for countable alphabets is given in Theorem~\ref{th:MarkovL} and its proof will be deferred to Section~\ref{se:proof}.  In this paper, we consider only those $P_{XYZ}$ satisfying $H(P_{XYZ}) < \infty$ and
\begin{eqnarray}
\sum_x p(x|y) \left( \log p(x|y) \right)^2 < C \label{eq:bounded}
\end{eqnarray}
for $y \in \calY$, where $C$ is finite.  These assumptions enable us to simplify the proofs by using Chebyshev's inequality.

\medskip
\begin{theorem} \label{th:MarkovL}
Consider $P_{XYZ}$ with $H(P_{XYZ}) < \infty$.  Assume that \eqref{eq:bounded} is satisfied and $X - Y -  Z$. If for any $\gamma > 0$ and any given $(\by,
\bz) \in U^n_{[YZ]\eta}$, $\bX$ is drawn $\sim \prod_{i}
p(x_i|y_i)$, then
\begin{eqnarray}
\Pr\big\{(\bX, \by, \bz) \in U^n_{[XYZ]\gamma} \big\} \geq 1 -
\gamma
\end{eqnarray}
for $n$ sufficiently large and $\eta$ sufficiently small.
\end{theorem}

\medskip \noindent\textbf{Remarks:}
\begin{itemize}
    \item[i)] This is a generalization of
        \cite[Lemma~15.8.1]{bk:Cover}. Since unified typicality
        retains the asymptotic equipartition property and the
        structural properties of strong typicality
        \cite{UnifiedTypicality}\cite{UnifiedTypicalityJ}, it is readily
        to generalize the achievability parts of Theorem~15.8.1 and
        Theorem~15.9.1 in \cite{bk:Cover} with $X$ and $Y$
        taking values from countable alphabets.
    \item[ii)] A result similar to \cite[(1.27)]{bk:Kramer} with
        strong typicality replaced by unified typicality can be easily
        shown from Theorem~\ref{th:MarkovL}.
    \item[iii)] Theorem~\ref{th:MarkovL} can easily generalize the
        version of the Markov lemma in \cite{Berger} to countably
        infinite alphabet as follows.
\end{itemize}

\medskip
\begin{corollary} \label{co:MarkovL2}
Consider $P_{XYZ}$ with $H(P_{XYZ}) < \infty$.  Assume that \eqref{eq:bounded} is satisfied and $X - Y -  Z$. If for any $\gamma > 0$ and any given $\bz \in U^n_{[Z]\eta}$, $(\bX, \bY)$ is generated according to
$\Pr\{(\bX, \bY) = (\bx, \by) \} = \prod_{i} p(x_iy_i)$, then
\begin{eqnarray}
\Pr\big\{(\bX, \bz) \in U^n_{[XZ]\gamma} | (\bY, \bz) \in
U^n_{[YZ]\eta} \big\} \geq 1 - \gamma,
\end{eqnarray}
for $n$ sufficiently large and $\eta$ sufficiently small.
\end{corollary}
\medskip
\begin{proof}
If $(\bX, \bY, \bz) \in U^n_{[XYZ]\gamma}$, then $(\bX, \bz) \in
U^n_{[XZ]\gamma}$ from the consistency theorem in
\cite[Theorem 5]{UnifiedTypicality}. Therefore,
\begin{eqnarray}
\lefteqn{\Pr\big\{(\bX, \bz) \in U^n_{[XZ]\gamma}| (\bY, \bz) \in
U^n_{[YZ]\eta} \big\}} \\
&\geq& \Pr\big\{(\bX, \bY, \bz) \in U^n_{[XYZ]\gamma}| (\bY, \bz) \in
U^n_{[YZ]\eta} \big\}  \label{eq:MarkovLv20}\\
&=& \sum_{\by: (\by, \bz) \in U^n_{[YZ]\eta}} \Pr\{\bY = \by| (\by,
\bz) \in U^n_{[YZ]\eta}\} \cdot \nonumber\\
&& \Pr\big\{(\bX, \by, \bz) \in U^n_{[XYZ]\gamma}| (\by,
\bz) \in U^n_{[YZ]\eta} \big\} \\
&\geq& 1-\gamma, \label{eq:MarkovLv21}
\end{eqnarray}
where \eqref{eq:MarkovLv21} follows from
Theorem~\ref{th:MarkovL}.
\end{proof}


\medskip
\subsection{Some Lemmas} \label{se:somelemmas}
In order to prove Theorem~\ref{th:MarkovL}, we have to first establish the results in this subsection.  Let $E_i$ be an events for all $i$.  In this paper, we will frequently use the following lemma and the fact that if $E_1$ implies $E_2$, then $\Pr\{E_1\} \leq \Pr\{E_2\}$.

\medskip
\begin{lemma} \label{lm:union}
If $\Pr\{E_i\} \geq 1 - \delta_i$, then
\begin{eqnarray}
\Pr\{\cap_i E_i\} \geq 1 - \sum_i \delta_i.
\end{eqnarray}
\end{lemma}
\medskip
\begin{proof} By the union bound,
\begin{eqnarray*}
\Pr\{\cap_i E_i\} = 1 - \Pr\{\cup_i E_i^c\} \geq 1 - \sum_i
\Pr\{E_i^c\} \geq 1 - \sum_i \delta_i.
\end{eqnarray*}
\end{proof}

In the following lemma, we consider the variational distance (see
e.g., \cite{bk:Raymond}) between $Q_{XYZ}$ and $P_{XYZ}$ which
is defined as
\begin{eqnarray}
V(Q_{XYZ}, P_{XYZ}) = \sum_{xyz} |q(xyz) - p(xyz)|.
\end{eqnarray}

\medskip

\begin{lemma}\label{lm:strongpart_pxyz}
Assume $X - Y -  Z$. If for any $\epsilon > 0$ and any given $(\by,
\bz) \in U^n_{[YZ]\eta}$, $\bX$ is drawn $\sim \prod_{i}
p(x_i|y_i)$, then
\begin{eqnarray}
\Pr\left\{V(Q_{XYZ}, P_{XYZ}) \leq \epsilon \right\} \geq
1 - \epsilon
\end{eqnarray}
for $n$ sufficiently large and $\eta$ sufficiently small.
\end{lemma}
%
\medskip
\begin{proof}
The proof is similar to the proof of \cite[Lemma~4.1]{Berger}
except that $P_{XYZ}$ is defined on countable alphabets here. Fix
any  $(x,y,z) \in \calX \times \calY \times \calZ$. For $1 \leq i \leq
n$, let $B_i$ be binary and independently distributed. If $(y,z) =
(y_i, z_i)$, let
\begin{eqnarray}
B_i = \left\{
        \begin{array}{ll}
          0 & \mbox{ with probability } 1 - p(x|y)\\
          1 & \mbox{ with probability } p(x|y).
        \end{array}
    \right. 
\end{eqnarray}
If $(y,z) \neq (y_i, z_i)$, let $B_i = 0$. Then $N(x,y,z;\bX,
\by, \bz)$ and $\sum_{i=1}^n B_i$ have the same distribution on
the set of integers.  So
\begin{eqnarray}
\bE[N(x,y,z;\bX, \by, \bz)] = \sum_{i=1}^n \bE[B_i]= p(x|y)
N(y,z;\by, \bz).
\end{eqnarray}
Since $B_i$ are binary and independent, the variance of
$N(x,y,z;\bX, \by, \bz)$ is
\begin{eqnarray}
\mbox{Var}[N(x,y,z;\bX, \by, \bz)] = \sum_{i=1}^n \mbox{Var}[B_i] \leq n.
\end{eqnarray}
For any $\delta > 0$, Chebyshev's inequality \cite[(3.32)]{bk:Cover}
can be applied to show
\begin{eqnarray}
{\Pr\{|N(x,y,z;\bX, \by, \bz) - p(x|y) N(y,z;\by, \bz)| \geq
n\delta\} } \nonumber\\
\leq \frac{\mbox{Var}[N(x,y,z;\bX, \by, \bz)]}{(n\delta)^2} \leq
\frac{1}{n \delta^2} \leq \delta, \ \ \ \ \ \ \ \ \ \ \ \
\label{eq:strp1}
\end{eqnarray}
where the last inequality holds for sufficiently large $n$.  Since
$q(xyz) = n^{-1}N(x,y,z;\bX, \by, \bz)$ and $q(yz) =
n^{-1}N(y,z;\by, \bz)$, \eqref{eq:strp1} is equivalent to
\begin{eqnarray}
\Pr\{|q(xyz) - p(x|y) q(yz)| \leq \delta\} \geq 1 - \delta.
\label{eq:lmS-1}
\end{eqnarray}

Now for any $\epsilon > 0$, let
\begin{eqnarray}
\eta = \frac{\epsilon^2}{32}. \label{eq:lmS3}
\end{eqnarray}
Since $(\by, \bz) \in U^n_{[YZ]\eta}$, $D(Q_{YZ} || P_{YZ}) \leq
\eta = \frac{\epsilon^2}{32}$. By Pinsker's inequality \cite{bk:Cover} and the fact that $\ln 2 < 1$,
\begin{eqnarray}
\frac{\epsilon}{4} &\geq& \sum_{yz} \left|q(yz) - p(yz)\right| \\
&=& \sum_{xyz} p(x|y)\left|q(yz) - p(yz)\right| \\
&=& \sum_{xyz} \left|p(x|y)q(yz) - p(xyz)\right|, \label{eq:strp2}
\end{eqnarray}
where \eqref{eq:strp2} follows from that $X - Y -  Z$.
Let $M = |\calS|$ where $\calS \subset \calX\times\calY\times\calZ$
is a finite subset such that
\begin{eqnarray}
\sum_{(x,y,z) \in \calS} p(xyz) \geq 1 - \frac{\epsilon}{8}.
\label{eq:minM}
\end{eqnarray}
Here, the left side of \eqref{eq:minM} goes to 1 as $M \rightarrow
\infty$, so that such $\calS$ must exist.  Let $E_{xyz} =
\b1\{|q(xyz) - p(x|y) q(yz)| \leq \frac{\epsilon}{8M}\}$ and suppose
$E_{xyz} = 1$ for all $(x,y,z) \in \calS$. Then
\begin{eqnarray}
\sum_{(x,y,z) \in \calS} |q(xyz) - p(x|y) q(yz)| \leq \frac{\epsilon}{8}.
\label{eq:lmS-0.5}
\end{eqnarray}
Together with \eqref{eq:strp2}, we have
\begin{eqnarray}
\sum_{(x,y,z) \in \calS} |q(xyz) - p(xyz)| \leq \frac{3\epsilon}{8}. \label{eq:lms4}
\end{eqnarray}
and hence,
\begin{eqnarray}
\sum_{(x,y,z) \in \calS} q(xyz) \geq \sum_{(x,y,z) \in \calS} p(xyz) -
\frac{3\epsilon}{8} \geq 1 - \frac{\epsilon}{2}, \label{eq:lmS0}
\end{eqnarray}
where the last inequality follows from \eqref{eq:minM}. Thus,
\begin{eqnarray}
\lefteqn{\sum_{xyz} |q(xyz) - p(xyz)|}\\
&\leq& \sum_{(x,y,z) \in \calS} |q(xyz) - p(xyz)| + \left(1 - \sum_{(x,y,z) \in \calS} q(xyz)\right)\nonumber\\
&&  +
\left(1 - \sum_{(x,y,z) \in \calS} p(xyz)\right)\\
&\leq& \frac{3\epsilon}{8} + \frac{\epsilon}{2} +
\frac{\epsilon}{8} \label{eq:lmS2}\\
&=& \epsilon,
\end{eqnarray}
where \eqref{eq:lmS2} follows from \eqref{eq:minM},
\eqref{eq:lms4} and \eqref{eq:lmS0}. Therefore, if $E_{xyz} = 1$
for all $(x,y,z) \in \calS$, then
$V(Q_{XYZ}, P_{XYZ}) \leq \epsilon.$
So we can put $\delta = \frac{\epsilon}{8M}$ into \eqref{eq:lmS-1}
and apply Lemma~\ref{lm:union} to show that when $n$ is
sufficiently large,
\begin{eqnarray}
\Pr \left\{V(Q_{XYZ}, P_{XYZ}) \leq \epsilon
\right\} &\geq& \Pr\{\cap_{(x,y,z) \in \calS} E_{xyz}\} \nonumber\\
&\geq& 1 - \frac{\epsilon}{8} \geq 1-\epsilon. 
\end{eqnarray}
\end{proof}
\medskip

We now establish a result regarding the Kullback-Leibler
divergence and entropy difference between $P_{X|YZ}$ and
$Q_{X|YZ}$.  In the following lemma, $(y^n, z^n)$ is not necessarily jointly typical.  Also, $Q_{X|Y=y, Z=z}$ and $P_{X|Y=y, Z=z}$
are the probability distributions of $X$ when $Y=y$ and $Z=z$ are given.  Recall that we consider only those $P_{XYZ}$ satisfying \eqref{eq:bounded} and $H(P_{XYZ}) < \infty$.

\medskip
\begin{lemma} \label{lm:weakpart}
Assume $X - Y -  Z$. If for any $\epsilon
> 0$ and any given $(y^n, z^n)$, $\bX$ is drawn $\sim \prod_{i} p(x_i|y_i)$, then
\begin{eqnarray}
\lefteqn{\Pr\left\{\bigg|\sum_{yz} q(yz) \left(D(Q_{X|Y=y, Z=z}||P_{X|Y=y, Z=z}) + \right.\right.} \nonumber\\
&& \left.\left.H(Q_{X|Y=y,Z=z}) - H(P_{X|Y=y,Z=z}) \right)\bigg|
\leq \epsilon \right\} \geq 1 - \epsilon \nonumber\\
\label{eq:weakpart}
\end{eqnarray}
for $n$ sufficiently large.
\end{lemma}
\begin{proof}
For $1 \leq i \leq n$, let $A_i = \log p(X_i|y_i)$. Since $X_i$ are
independent, $A_i$ are also independent. Together with \eqref{eq:bounded}, the upper bound on the variance of $\sum_{i=1}^n A_i$ is given by
\begin{eqnarray}
\mbox{Var} \left[\sum_{i=1}^n A_i \right] = \sum_{i=1}^n \mbox{Var}[ A_i]
\leq \sum_{i=1}^n\bE[A_i^2]
\leq nC,
\end{eqnarray}
By Chebyshev's inequality,
\begin{eqnarray}
\Pr\left\{\left|\sum_{i=1}^n A_i - \bE \left[\sum_{i=1}^n A_i \right] \right| \geq n \epsilon\right\} &\leq& \frac{\mbox{Var} \left[\sum_{i=1}^n A_i \right]}{(n \epsilon)^2} \nonumber\\
&\leq& \frac{C}{n \epsilon^2} \leq \epsilon 
\end{eqnarray}
when $n$ is sufficiently large.  Then
\begin{eqnarray}
{\Pr\left\{\left|n^{-1}\sum_{i=1}^n A_i - n^{-1} \bE \left[\sum_{i=1}^n A_i \right] \right| \leq \epsilon\right\}}
\geq 1 - \epsilon, \label{eq:weakp5}
\end{eqnarray}
where the left sides of \eqref{eq:weakpart} and \eqref{eq:weakp5}
are equal because
\begin{eqnarray}
\lefteqn{ n^{-1}\bE \left[\sum_{i=1}^n A_i
\right] - n^{-1}\sum_{i=1}^n A_i } \\
&=& n^{-1} \sum_{i=1}^n\sum_{x} p(x|y_i) \log p(x|y_i) - n^{-1} \sum_{i=1}^n \log p(X_i|y_i) \nonumber\\
&=& n^{-1} \sum_{y} N(y; \by) \sum_{x}p(x|y) \log p(x|y) \nonumber\\
&& - n^{-1} \sum_{xy} N(x,y;\bX, \by) \log p(x|y) \\
&=&\sum_{xy} p(x|y)q(y) \log p(x|y) - \sum_{xy} q(x,y) \log p(x|y) \\
&=& \sum_{xyz} (p(x|y)q(yz) - q(xyz)) \log p(x|y) \\
&=& \sum_{yz} q(yz) \sum_x (p(x|yz) - q(x|yz)) \log p(x|yz) \label{eq:weakp3}\\
&=& \sum_{yz} q(yz) \sum_x \bigg(q(x|yz)\log \frac{q(x|yz)}{p(x|yz)} - \nonumber\\
&&q(x|yz) \log q(x|yz) + p(x|yz) \log p(x|yz) \bigg), \label{eq:weakp4}
\end{eqnarray}
where \eqref{eq:weakp3} follows from that $X - Y -  Z$.
\end{proof}

If $(\by, \bz) \in U^n_{[YZ]\eta}$, the following lemma simplifies
\eqref{eq:weakpart}.
\medskip
\begin{lemma} \label{lm:weakpart_pyz}
For any $\epsilon > 0$, there exists $\eta > 0$ such that if $(\by,
\bz) \in U^n_{[YZ]\eta}$, then
\begin{eqnarray}
\left|\sum_{yz} (q(yz) - p(yz)) H(P_{X|Y=y,Z=z}) \right| \leq \epsilon, \label{eq:weakpyz0}
\end{eqnarray}
where $\epsilon \rightarrow 0$ as $\eta \rightarrow 0$.
\end{lemma}
\medskip
\begin{proof}
Since
\begin{eqnarray*}
{\sum_x p(x|y) \left( \log p(x|y) \right)^2} \ \ \geq \ \sum_{x : p(x|y) > 0.5} p(x|y) \left( \log p(x|y) \right)^2 \nonumber\\
\  - \sum_{x : p(x|y) \leq 0.5} p(x|y) \left( \log p(x|y) \right),
\end{eqnarray*}
it is easily shown that $H(P_{X|Y=y}) \leq 0.5 + C$ from
\eqref{eq:bounded}. Since $p(x|yz) = p(x|y)$ for all $(x,y,z)$ as $X
- Y -  Z$,
\begin{eqnarray}
\lefteqn{\left|\sum_{yz} (q(yz) - p(yz)) H(P_{X|Y=y,Z=z}) \right|} \\
&=&\left|\sum_{yz} (q(yz) - p(yz)) H(P_{X|Y=y}) \right| \label{eq:weakpyz0.5}\\
&\leq& \sum_{yz: q(yz) \geq p(yz)} (q(yz) - p(yz)) (0.5 + C) + \nonumber\\
&& \sum_{yz: q(yz) < p(yz)} (p(yz) - q(yz)) (0.5 + C) \label{eq:weakpyz1}\\
&=& (0.5 + C)\sum_{yz} |p(yz) - q(yz)| ,  \\
&\leq& (0.5 + C)\sqrt{2 \eta \ln 2 } , \label{eq:weakpyz2}
\end{eqnarray}
where \eqref{eq:weakpyz2} follows from $(\by, \bz) \in
U^n_{[YZ]\eta}$ and Pinsker's inequality. By letting $\eta =
\frac{\epsilon^2}{(0.5+C)^2 2\ln2}$, the lemma is proved.
\end{proof}

Now we use Lemma~\ref{lm:weakpart_pyz} to simplify
\eqref{eq:weakpart} in the following lemma, which uses the
conditional Kullback-Leibler divergence
$D(Q_{X|YZ}||P_{X|YZ}|Q_{YZ})$ \cite{bk:TSHanKKobay}.
\medskip
\begin{lemma} \label{lm:weakpart_pxyz}
Assume $X - Y -  Z$. If for any $\epsilon > 0$ and any given $(\by,
\bz) \in U^n_{[YZ]\eta}$, $\bX$ is drawn $\sim \prod_{i}
p(x_i|y_i)$, then
\begin{eqnarray}
\lefteqn{\Pr\big\{|D(Q_{X|YZ}||P_{X|YZ}|Q_{YZ}) + } \nonumber\\
&& H(Q_{X|YZ}) - H(P_{X|YZ})| \leq \epsilon \big\} \geq 1 -
\epsilon
\end{eqnarray}
for $n$ sufficiently large and $\eta$ sufficiently small.
\end{lemma}
\medskip
\begin{proof}
For any $\epsilon > 0$, there exists a sufficiently small
$\eta$ such that
\begin{eqnarray}
\left|\sum_{yz} (q(yz) - p(yz)) H(P_{X|Y=y,Z=z}) \right| \leq \frac{\epsilon}{2} \label{eq:weakpxyz1}
\end{eqnarray}
from Lemma~\ref{lm:weakpart_pyz}.  Now, suppose
\begin{eqnarray}
\lefteqn{\bigg|\sum_{yz} q(yz) (D(Q_{X|Y=y, Z=z}||P_{X|Y=y, Z=z}) + }\nonumber\\
&& H(Q_{X|Y=y,Z=z}) - H(P_{X|Y=y,Z=z}) )\bigg|
\leq \frac{\epsilon}{2}.  \label{eq:weakpxyz2}
\end{eqnarray}
Adding \eqref{eq:weakpxyz1} and \eqref{eq:weakpxyz2} gives
\begin{eqnarray}
&&{|D(Q_{X|YZ}||P_{X|YZ}|Q_{YZ}) + H(Q_{X|YZ}) - H(P_{X|YZ})|} \nonumber\\
&& \ \ \ \ \ \leq \epsilon. \label{eq:weakpxyz3}
\end{eqnarray}
When $n$ is sufficiently large, the probability that
\eqref{eq:weakpxyz2} is satisfied is larger than $1 -
\frac{\epsilon}{2}
> 1 - \epsilon$ from
Lemma~\ref{lm:weakpart}. Therefore, the lemma is proved.
\end{proof}
\medskip

%
Before we process to apply the established lemmas, we pause to
check that conditional entropy similar to entropy is lower
semicontinuous. Let $P_{A_mB_m} = \{p_{A_mB_m}(ab)\}$ and
$P_{AB} = \{p_{AB}(ab)\}$. Assume $H(P_{A|B}) < \infty$.
\medskip
\begin{lemma} \label{lm:conditionalH}
If $\lim_{m \rightarrow \infty} V(P_{A_mB_m}, P_{AB}) = 0$, then
$\lim_{m \rightarrow \infty} H(P_{A_m|B_m}) \geq H(P_{A|B})$.
\end{lemma}
\medskip
\begin{proof}
For any $\epsilon > 0$, there exists sufficient large $L$ and $M$
such that
\begin{eqnarray}
H(P_{A|B}) \leq \sum_{b=1}^M p_{B}(b)\tilde{H}(P_{A|B=b}) + \epsilon,
\label{eq:chHPappro}
\end{eqnarray}
where
$\tilde{H}(P_{A|B=b}) = -\sum_{a=1}^L p_{A|B}(a|b) \log
p_{A|B}(a|b).$
On the other hand,
\begin{eqnarray}
H(P_{A_m|B_m}) &\geq& \sum_{b=1}^M p_{B_m}(b) H(P_{A_m|B_m=b})\\
&\geq& \sum_{b=1}^M p_{B_m}(b) \tilde{H}(P_{A_m|B_m=b})
\label{eq:chHQappro},
\end{eqnarray}
where the right side of \eqref{eq:chHQappro} is a continuous
function in $\{p_{A_mB_m}(ab): 1 \leq a \leq L \mbox{ and } 1 \leq
b \leq M\}$. If $\lim_{m \rightarrow \infty} V(P_{A_mB_m},
P_{AB}) = 0$, $p_{A_mB_m}(ab) \rightarrow p(ab)$ for all $1 \leq
a \leq L$ and $1 \leq b \leq M$. Following \eqref{eq:chHQappro},
by replacing $p_{A_mB_m}$ by $p_{AB}$ and $P_{A_m|B_m=b}$
by $P_{A|B=b}$ on the right side, for any $\epsilon
> 0$,
\begin{eqnarray}
\lim_{m \rightarrow \infty} H(P_{A_m|B_m}) &\geq& \sum_{b=1}^M p_B(b) \tilde{H}(P_{A|B=b}) - \epsilon\\
&\geq& H(P_{A|B}) - 2 \epsilon, \label{eq:chHQappro2}
\end{eqnarray}
where \eqref{eq:chHQappro2} follows from \eqref{eq:chHPappro}.
Since $\epsilon > 0$ is arbitrary, the lemma is proved.
\end{proof}

By Lemma~\ref{lm:strongpart_pxyz} and
Lemma~\ref{lm:conditionalH}, we are capable to strengthen
Lemma~\ref{lm:weakpart_pxyz} and give the following lemma.
\medskip
\begin{lemma} \label{lm:MarkovL_prepare}
Assume $X - Y -  Z$. If for any $\epsilon > 0$ and any given $(\by,
\bz) \in U^n_{[YZ]\eta}$, $\bX$ is drawn $\sim \prod_{i}
p(x_i|y_i)$, then
\begin{eqnarray}
\Pr\big\{D(Q_{X|YZ}||P_{X|YZ}|Q_{YZ}) \leq \epsilon \big\} \geq 1 -
\epsilon, \label{eq:Markovp1}
\end{eqnarray}
and
\begin{eqnarray}
\Pr\big\{|H(Q_{X|YZ}) - H(P_{X|YZ})| \leq \epsilon \big\} \geq 1 -
\epsilon \label{eq:Markovp2}
\end{eqnarray}
for $n$ sufficiently large and $\eta$ sufficiently small.
\end{lemma}
\medskip
\begin{proof}
For any $\epsilon > 0$ and $P_{XYZ}$, there exists a sufficiently
small $\delta$ from Lemma~\ref{lm:conditionalH} such that if
\begin{eqnarray}
V(Q_{XYZ}, P_{XYZ}) \leq \delta, \label{eq:MarkovL_prepare1}
\end{eqnarray}
then $H(Q_{X|YZ}) - H(P_{X|YZ}) \geq - \epsilon$. On the other
hand, if \eqref{eq:weakpxyz3} is satisfied, then $\epsilon \geq
H(Q_{X|YZ}) - H(P_{X|YZ})$.
Therefore, if both \eqref{eq:weakpxyz3} and
\eqref{eq:MarkovL_prepare1} are satisfied, then $|H(Q_{X|YZ}) -
H(P_{X|YZ})| \leq \epsilon$. When $n$ is sufficiently large and
$\eta$ is sufficiently small, Lemma~\ref{lm:strongpart_pxyz} shows
that
\begin{eqnarray}
\Pr\left\{V\left(Q_{XYZ}, P_{XYZ}\right) \leq \min\left\{\delta, \frac{\epsilon}{2}\right\} \right\} \geq
1 - \frac{\epsilon}{2}.
\end{eqnarray}
Also, Lemma~\ref{lm:weakpart_pxyz} shows that
\eqref{eq:weakpxyz3} is true with probability larger than $1 -
\frac{\epsilon}{2}$. Therefore, \eqref{eq:Markovp2} can be shown
from Lemma~\ref{lm:union}. Similarly, \eqref{eq:Markovp1} can
be verified by Lemma~\ref{lm:union},
Lemma~\ref{lm:weakpart_pxyz} together with
\eqref{eq:Markovp2}.

\end{proof}

Due to the following theorem, we just need to bound two instead of
eight quantities in \eqref{eq:UJtypicality} in order to verify that
$(\bx, \by, \bz) \in U_{[XYZ]\gamma}^n$.

\medskip
\begin{theorem} \label{th:disconBset}
Assume $H(P_{AB})$ is finite. If $\lim_{m \rightarrow \infty}
V(P_{A_mB_m}, P_{AB}) = 0$ and $\lim_{m \rightarrow \infty}
|H(P_{A_mB_m}) - H(P_{AB})| = 0$, then

\begin{eqnarray}
\lim_{m \rightarrow \infty} |H(P_{A_m}) - H(P_{A})| = 0.
\end{eqnarray}
\end{theorem}
\medskip
\begin{proof}
\begin{eqnarray}
\lim_{m \rightarrow \infty} H(P_{A_m}) &=& \lim_{m \rightarrow
\infty} H(P_{A_mB_m}) - H(P_{B_m | A_m}) \\
&=& H(P_{AB}) - \lim_{m \rightarrow
\infty} H(P_{B_m | A_m}) \\
&\leq& H(P_{AB}) - H(P_{B | A}) \label{eq:Hcon1}\\
&=& H(P_{A}),
\end{eqnarray}
where \eqref{eq:Hcon1} follows from
Lemma~\ref{lm:conditionalH}. On the other hand, $\lim_{m
\rightarrow \infty} H(P_{A_m}) \geq H(P_{A})$ because entropy is
lower semi-continuous \cite{topsoeBC}. Therefore, the theorem is
proved.
\end{proof}
\medskip
%
Suppose $|H(Q_{XYZ})-H(P_{XYZ})|$ and $D(Q_{XYZ}||P_{XYZ})$
are sufficiently small.  In this case, $V(Q_{XYZ}, P_{XYZ})$ is small
from Pinsker's inequality and Theorem~\ref{th:disconBset} tells that
all the nonnegative quantities in \eqref{eq:UJtypicality} are also
small.

\subsection{Proof of Theorem~\ref{th:MarkovL} \label{se:proof}}
We first show that for any $\epsilon > 0$,
\begin{eqnarray}
\Pr\big\{|H(Q_{XYZ}) - H(P_{XYZ})| \leq \epsilon \big\} \geq 1 -
\frac{\epsilon}{2}, \label{eq:Markovl2}
\end{eqnarray}
and
\begin{eqnarray}
\Pr\big\{D(Q_{XYZ}||P_{XYZ}) \leq \epsilon \big\} \geq 1 -
\frac{\epsilon}{2} \label{eq:Markovl1}
\end{eqnarray}
when $n$ is sufficiently large and $\eta$ is sufficiently
small.


Let $\eta = \frac{\epsilon}{2}$ so that $|H(Q_{YZ}) - H(P_{YZ})|
\leq \frac{\epsilon}{2}$ as $(\by, \bz) \in U^n_{[YZ]\eta}$. If
$|H(Q_{X|YZ}) - H(P_{X|YZ})| \leq \frac{\epsilon}{2}$, then
\begin{eqnarray}
\epsilon &\geq& |H(Q_{X|YZ}) - H(P_{X|YZ})| + |H(Q_{YZ}) - H(P_{YZ})| \nonumber\\
&\geq& |H(Q_{XYZ}) - H(P_{XYZ})|. \label{eq:mlH}
\end{eqnarray}
Together with Lemma~\ref{lm:MarkovL_prepare},
\eqref{eq:Markovl2} follows from
\begin{eqnarray}
\lefteqn{\Pr\left\{|H(Q_{XYZ}) - H(P_{XYZ})| \leq \epsilon
\right\}} \nonumber\\
&\geq&\Pr\left\{|H(Q_{X|YZ}) - H(P_{X|YZ})| \leq \frac{\epsilon}{2}
\right\} \geq 1 - \frac{\epsilon}{2}. \ \ \ \ \ \
\end{eqnarray}


Since $\eta = \frac{\epsilon}{2}$, $D(Q_{YZ} || P_{YZ}) \leq
\frac{\epsilon}{2}$ as $(\by, \bz) \in U^n_{[YZ]\eta}$.  If
$D(Q_{X|YZ}||P_{X|YZ} | Q_{YZ}) \leq \frac{\epsilon}{2}$, then
\begin{eqnarray}
\epsilon &\geq& D(Q_{X|YZ}||P_{X|YZ} | Q_{YZ}) + D(Q_{YZ} || P_{YZ})\\
&=& D(Q_{XYZ}||P_{XYZ}). \label{eq:mlD}
\end{eqnarray}
Together with Lemma~\ref{lm:MarkovL_prepare},
\eqref{eq:Markovl1} follows from
\begin{eqnarray}
\lefteqn{\Pr\left\{D(Q_{XYZ}||P_{XYZ}) \leq \epsilon
\right\}}\\
&\geq& \Pr\left\{D(Q_{X|YZ}||P_{X|YZ}|Q_{YZ}) \leq \frac{\epsilon}{2}
\right\} \geq 1 - \frac{\epsilon}{2}. \ \ \
\end{eqnarray}

For any $\gamma > 0$, there exists a sufficiently small $\epsilon
\leq \frac{\gamma}{8}$ from Theorem~\ref{th:disconBset} such
that if \eqref{eq:mlH} and \eqref{eq:mlD} are satisfied, then all the
absolute values in \eqref{eq:UJtypicality} are less than
$\frac{\gamma}{8}$, and hence, \eqref{eq:UJtypicality} is satisfied.
Therefore, by \eqref{eq:Markovl2} and \eqref{eq:Markovl1},
\begin{eqnarray}
\lefteqn{\Pr\big\{(\bX, \by, \bz) \in U^n_{[XYZ]\gamma} \big\}} \nonumber\\
&\geq& \Pr\big\{\{|H(Q_{XYZ}) - H(P_{XYZ})| \leq \epsilon\} \mbox{ and } \nonumber\\
&&\{D(Q_{XYZ}||P_{XYZ}) \leq \epsilon \}\big\} \\
&\geq& 1 - \epsilon \\
&\geq& 1 - \gamma.
\end{eqnarray}
{\hspace*{\fill}~\QED\par\endtrivlist\unskip}
%
%
%

\section{Conclusion}
A version of the Markov lemma which works on both finite or countably
infinite alphabets has been proved.  We have also demonstrated a
method to ease the verification of jointly unified typical sequences.
These results can readily generalize the achievability parts in some
existing coding theorems to countably infinite alphabet and they are
potentially useful for proving coding theorems that apply to both
finite and infinite alphabets.

\end{document}